\documentclass[letterpaper, 10 pt, conference]{ieeeconf}  

\IEEEoverridecommandlockouts                              

\usepackage{amsmath}
\usepackage{amssymb}
\usepackage{algorithm}
\usepackage{algorithmic}

\usepackage{graphicx}
\usepackage{subcaption}
\usepackage{caption}

\usepackage{booktabs}

\usepackage{makecell}
\usepackage{siunitx}
\newtheorem{theorem}{Theorem}
\newtheorem{proposition}[theorem]{Proposition}
\newtheorem{assumption}[theorem]{Assumption}

\title{\LARGE \bf
FEDBUD: Joint Incentive and Privacy Optimization \\ for Resource-Constrained Federated Learning
}

\author{Tao Liu$^{1}$ and Xuehe Wang$^{1,*}$
\thanks{$^{1}$Tao Liu and Xuehe Wang are with the School of Artificial Intelligence, Sun Yat-sen University, Zhuhai 519082, China (e-mail: {liut353@mail2.sysu.edu.cn; wangxuehe@mail.sysu.edu.cn}).}
\thanks{$^{*}$Corresponding author: Xuehe Wang.}%
}

\begin{document}

\maketitle
\thispagestyle{empty}
\pagestyle{empty}

\begin{abstract}

Federated learning has become a popular paradigm for privacy protection and edge-based machine learning. However, defending against differential attacks and devising incentive strategies remain significant bottlenecks in this field. Despite recent works on privacy-aware incentive mechanism design for federated learning, few of them consider both data volume and noise level. In this paper, we propose a novel federated learning system called FEDBUD, which combines privacy and economic concerns together by considering the joint influence of data volume and noise level on incentive strategy determination. In this system, the cloud server controls monetary payments to edge nodes, while edge nodes control data volume and noise level that potentially impact the model performance of the cloud server. To determine the mutually optimal strategies for both sides, we model FEDBUD as a two-stage Stackelberg Game and derive the Nash Equilibrium using the mean-field estimator and virtual queue. Experimental results on real-world datasets demonstrate the outstanding performance of FEDBUD.

\end{abstract}

\section{INTRODUCTION}
With the development of the Internet of Things, many smart things (mobile phones, wearable devices, electric vehicles) are generating a large amount of data every day. The traditional machine learning paradigm of uploading data from edge nodes to cloud server for centralized model training faces key challenges: on the one hand, it is incapable of taking advantage of growing storage and computational power on edge nodes; on the other hand, straight data transfer between the cloud server and edge nodes may incur malicious privacy attack, thereby lead to privacy leakage over data owner.
In response to these concerns, federated learning provides a solution by supporting edge nodes to train a model locally before uploading it to the cloud server for model aggregation.
This framework leverages edge computation resources while protecting data privacy effectively, and has been widely applied in various scenarios such as smart cities \cite{gandhi2025federated} and smart healthcare \cite{nasajpour2025federated}.

Despite the above advantages, federated learning still faces two bottlenecks: 1) \textit{Differential Attack}: It is a privacy inference technology that enables attackers to infer sensitive information from model parameters. Vanilla federated learning is unable to defend data privacy against it amid model transmission; 2) \textit{Resource Expenditure}: From an economic aspect, edge nodes inevitably consume computation and communication resources when performing model training and uploading. Without enough economic reward, they may be reluctant to participate in federated learning tasks.

For the first bottleneck, researchers have proposed a widely used framework called differential privacy, which enhances the ability of federated learning to defend against malicious attacks by injecting tunable levels of noise into the local model before it is uploaded.
In addition, variants of differential privacy have been developed for specific concerns, including data distribution \cite{yang2015bayesian, triastcyn2020bayesian} and information aging \cite{zhang2023age, lin2024age}. However, few studies investigate differential privacy from an economic optimization perspective.

For the second bottleneck, efforts have been made in incentive mechanism design where the cloud server provides elaborate monetary payment to stimulate edge nodes to participate in federated learning. It ranges from game theory \cite{huang2024collaboration, tang2025game} to auction theory \cite{chen2025dualgfl, tang2025reputation} and contract theory \cite{xie2025privacy}. However, most of them do not involve differential privacy in incentive mechanism design. Some works have proposed a privacy-aware incentive mechanism. Yet, they fail to account for the joint effect of data volume and noise level on payment strategy determination.

Motivated by the above discussion, this paper devises an innovative federated learning system called FEDBUD, which combines privacy and economic concern together by considering the joint influence of data volume and noise level on strategy determination. Specifically, the cloud server controls monetary payment to edge nodes while edge nodes control data volume and noise level that impact the model performance of the cloud server. The key questions in FEDBUD are:
1) \textit{Cloud server determines the optimal payment strategy to balance monetary payment to edge nodes and model performance influenced by edge nodes}.
2) \textit{Edge nodes determine their optimal data volume and noise level simultaneously to balance the resource cost and allocated payment from the cloud server}.

There are three challenges to solve the key questions:
1) \textit{Absence of model performance characterization}. Although model performance is influenced by data volume and noise level, there is a lack of a quantitative relationship linking these factors to model performance, which hinders the strategy determination for the cloud server.
2) \textit{Incomplete information}. Edge nodes' strategies are interdependent, as the allocated payment is based on relative contribution. But in federated learning, the phenomenon of information silo among edge nodes makes individual strategy optimization challenging. 
3) \textit{Resource Constraints}. Computation and communication resources of each edge node are limited during a federated learning task. How to allocate resources for each round to optimize the long-term objective is non-trivial.

To overcome the above challenges and determine the optimal strategies for both sides, we conduct a theoretical analysis for FEDBUD, model a two-stage Stackelberg Game, and derive the Nash Equilibrium using the mean-field estimator and virtual queue. The main contributions in this paper are summarized as follows: 

\begin{itemize}
    \item We propose an innovative federated learning system called FEDBUD, which combines privacy and economic concern together by considering the joint influence of data volume and noise level on strategy determination.
    \item We conduct a theoretical analysis on FEDBUD and uncover a quantitative tie linking model performance with data volume and noise level.
    \item We model FEDBUD as a two-stage Stackelberg game. By means of backward reduction, we explore the optimal strategies of both edge nodes and cloud server using the mean-field estimator and virtual queue.
    \item We conduct experiments on real-world datasets to validate the viability and efficiency of FEDBUD compared with other benchmarks.
\end{itemize}

\section{Problem Formulation and Analysis}
\subsection{Federated Learning with Privacy Protection}
A typical federated learning system comprises a cloud server and $N$ edge nodes with $T$ communication rounds. In the system, the cloud server contains a global model with parameters $\mathbf w$ and edge node $k \in N$ holds a set of privacy data of $\mathcal{B}_k^t$ with volume of $B_k^t = |\mathcal{B}_k^t|$ at round $t$. 
The loss function of edge node $k$ based on global model parameters $\mathbf w$ is defined as 
\begin{equation}
F_k(\mathbf w ; \mathcal{B}_k^t) = \frac{1}{B_k^t}\sum_{j=1}^{B_k^t} f(\mathbf w; x_k^j, y_k^j),
\end{equation}
where $f(w; x_k^j, y_k^j)$ is the loss function of each data point $\{x_k^j, y_k^j\}\in\mathcal{B}_k^t$. 

During each round, the cloud server distributes its global model parameters $\mathbf{w}^t$ to edge nodes. Then, edge node $k$ performs a local model update based on its own data by
\begin{equation}
\label{formu:local}
\mathbf w_k^{t+1} = \mathbf w^t - \eta\nabla F_k(\mathbf w^t ; \mathcal{B}_k^t) + \mathbf n_k^t,    
\end{equation}
where $\eta$ is the learning rate and $\nabla F_k(\mathbf w^t ; \mathcal{B}_k^t)$ is the loss gradient of nodes $k$ at round $t$. $\mathbf n_k^t \sim N(0, {\sigma_k^t}^2)$ is noise injected into the local model for data privacy, where $\sigma_k^t = \frac{\eta C}{B_k^t \varepsilon_k^t}$ \cite{huang2024collaboration}. $C$ is a constant and $\varepsilon_k^t$ is the privacy budget. The noise level can be manipulated through privacy budget $\varepsilon_k^t$ set by edge node $k$.

Until each node completes local training and uploads local model parameter $\mathbf{w}_k^t$ to the cloud server, it will aggregate them by
\begin{equation}
\label{formu:aggre}
\mathbf{w}^{t+1} = \sum_{k=1}^N \frac{B_k^t}{\sum_{i=1}^N B_i^t} \mathbf{w}_k^{t+1}.
\end{equation}

Subsequently, the cloud server launches a new global model $\mathbf{w}^{t+1}$ to each edge node for the next round's training.

The goal of federated learning is to find the optimal model parameters $\mathbf w^*$ to minimize the global loss function, which is represented as
\begin{equation}
\mathbf w^* = \arg \min_\mathbf w F(\mathbf w) = \sum_{k=1}^N \frac{B_k^t}{\sum_{i=1}^N B_i^t} F_k(\mathbf w; \mathcal{B}_k^t),    
\end{equation}
where $\mathcal B^t = \cup_{k=1}^N \mathcal B_k^t$ is total data used for model training at round $t$.

\subsection{Convergence Analysis for Federated Learning with Privacy Protection}
Convergence analysis for model performance is provided in this section. In practice, it is challenging to derive accurate model performance in closed form. Therefore, we approximate it with a convergence upper bound, which takes into consideration the impact of data volume and privacy budget on model performance.
Before that, we introduce some assumptions on the local loss function $F_k(\mathbf w; \mathcal{B}_k^t)$, which have been widely used in previous work \cite{wang2023federated, luo2024adaptive}.

\begin{assumption}
\label{ass:1}
For $k \in \{1, ..., N\}, t \in \{0, \cdots, T-1\}$, $F_k(\mathbf w; \mathcal{B}_k^t)$ is $\rho-$Lipschitz, i.e., $\forall \mathbf w_1, \mathbf w_2, F_k(\mathbf w_1 ; \mathcal B_k^t) - F_k(\mathbf w_2; \mathcal B_k^t) \leq \rho \Vert \mathbf w_1 - \mathbf w_2 \Vert_2$.
\end{assumption}

\begin{assumption}
\label{ass:2}
For $k \in \{1, ..., N\}, t \in \{0, \cdots, T-1\}$, $F_k(\mathbf w; \mathcal{B}_k^t)$ is $\mu-$strong convex, i.e., $\forall \mathbf w$, $F_k(\mathbf w ; \mathcal B_k^t)$ satisfies $F_k(\mathbf w ; \mathcal B_k^t) - F_k(w^*) \leq \frac{1}{2\mu} \Vert \nabla F_k(\mathbf w ; \mathcal B_k^t) \Vert_2^2$.
\end{assumption}

\begin{assumption}
\label{ass:3}
For $k \in \{1, ..., N\}, t \in \{0, \cdots, T-1\}$, non-iid degree is bounded, i.e., $\Vert \nabla F_k(\mathbf w; \mathcal{B}_k^t) - \nabla F(\mathbf w) \Vert_2 \leq \lambda_k^t$.
\end{assumption}

\begin{assumption}
\label{ass:4}
For $k \in \{1, ..., N\}, t \in \{0, \cdots, T-1\}$, $\mathbf n_k^t$ is zero-mean and variance-bounded, i.e., $\mathbf n_k^t \sim N(0, {\sigma_k^t}^2)$ with $\sigma_k^t = \frac{\eta C}{B_k^t \varepsilon_k^t}$.
\end{assumption}

Then, the convergence analysis is given as follows.

\begin{theorem}
\label{the:1}
Under Assumptions \ref{ass:1}-\ref{ass:4}, with $\eta \leq \frac{1}{\rho}$, the convergence upper bound after $T$ rounds of global training can be formulated as 
\begin{align}
    & E[F(\mathbf w^T) - F(\mathbf w^*)] \notag \\
    \leq & {\kappa_1}^T E[F(\mathbf w^0) - F(\mathbf w^*)] \notag \\
    &+ \sum_{t=0}^{T-1} {\kappa_1}^{T-1-t} \left(\kappa_2\sum_{k=1}^N \frac{B_k^t}{B^t} \lambda_k^t + 
    \kappa_3 \sum_{k=1}^N \frac{\eta^2 C^2}{{B^t}^2 {\varepsilon_k^t}^2}\right),
\end{align}
where $\kappa_1 = 1 + 2\mu\rho \eta^2 - 2\mu\eta, \kappa_2 = \rho\eta^2, \kappa_3 = \frac{\rho d}{2}$.
\end{theorem}

The detailed proof is provided in Appendix A.1 in the supplementary material.
Equ. (\ref{the:1}) unveils that model performance in FedBUD is influenced by total data volume and privacy budget simultaneously. Apparently, the greater both the total data volume and privacy budget, the better the global model performs.

\section{Game Formulation}
In this section, we formulate a cost optimization problem for the cloud server and a utility optimization problem for each edge node, respectively. Afterwards, we formulate the potential interaction between optimization problems on both sides as a two-stage Stackelberg Game.

\subsection{Cost Optimization of Cloud Server}
\label{sect:1}
The cost of the cloud server consists of two units: accuracy loss of model performance, and monetary payment to edge nodes. Although it is hard to secure the exact form of accuracy loss, we approximate it with the convergence upper bound provided in Equ. (\ref{the:1}). Denote $R^t$ as payment to edge nodes at round $t$, the cost of the cloud server over the time horizon can be formulated as
\begin{equation}
\label{game:cost}
C(\boldsymbol R, \boldsymbol B, \boldsymbol \varepsilon) = \sum_{t=0}^{T-1}\left(\gamma_1 R^t + \sum_{k=1}^N \frac{{\kappa_1}^{T-1-t}\kappa_3\eta^2C^2}{{B^t}^2{\varepsilon_k^t}^2}\right),
\end{equation}
where $\boldsymbol R = \{R^t\}_{t=0}^{T-1}$, $\boldsymbol B = \{ \{B_k^t\}_{t=0}^{T-1}\}_{k=1}^N$, and $\boldsymbol \varepsilon = \{\{\varepsilon_k^t\}_{t=0}^{T-1}\}_{k=1}^N$. In addition, $\gamma_1 > 0$ is a factor to balance the influence between monetary payment and model accuracy loss. When $\gamma_1$ approaches 0, the cloud server prefers model performance enhancement rather than expenditure control.

The optimization problem on the cloud server's side can be formulated as
\begin{equation}
    \label{game:cost opti}
    \min_{\boldsymbol R} C(\boldsymbol R, \boldsymbol B, \boldsymbol \varepsilon).
\end{equation}

\subsection{Utility Optimization of Edge Nodes}
\label{sect:2}
For edge node $k$, the computation resource expenditure amid model training is associated with data volume $B_k^t$, while privacy risk expenditure amid model uploading is associated with privacy budget $\varepsilon_k^t$. We use $\alpha_k H_1(B_k^t)$ and $\beta_k H_2(\varepsilon_k^t)$ to quantify the two terms, respectively. $\alpha_k$ is the unit cost for computation resource, and $\beta_k$ is the unit cost for privacy risk. Both $H_1(\cdot)$ and $H_2(\cdot)$ are convex functions to capture the fact that an edge node's computation resource consumption and privacy risk increase convexly with the data volume $B_k^t$ and privacy risk $\varepsilon_k^t$, respectively. In this work, we choose the quadratic forms of $H_1(B_k^t) = (B_k^t)^2$ and $H_2(\varepsilon_k^t) = (\varepsilon_k^t)^2$, which has been widely adopted in expenditure formulation \cite{zhan2019free, nie2020multi}. Hence, the cost of edge node $k$ at round $t$ can be formulated as
\begin{equation}
    E_k(t) = \alpha_k(B_k^t)^2 + \beta_k(\varepsilon_k^t)^2.
\end{equation}

To stimulate edge nodes to provide high-quality local model parameters efficiently, the payment allocation strategy is formulated as
\begin{equation}
    P_k(t) = \max\left\{0, \frac{\log(B_k^t\varepsilon_k^t)}{\sum_{i=1}^N \log(B_i^t\varepsilon_i^t)} R^t\right\}.
\end{equation}
Under the above strategy, the payment edge node $k$ obtains at round $t$ depends on its data volume $B_k^t$ and privacy budget $\varepsilon_k^t$ compared with that of other edge nodes.

Therefore, the utility function of edge node $k$ over the time horizon is formulated as 
\begin{equation}
    \label{game:utility}
    U_k(\boldsymbol R, \boldsymbol B_k, \boldsymbol B_{-k}, \boldsymbol \varepsilon_k, \boldsymbol \varepsilon_{-k}) = \sum_{t=0}^{T-1} (P_k(t) - E_k(t)),
\end{equation}
where $\boldsymbol B_{-k} = \boldsymbol B \backslash \boldsymbol B_k$, and $\boldsymbol \varepsilon_{-k} = \boldsymbol \varepsilon \backslash \boldsymbol \varepsilon_k$.

In addition, in the real world, computation resource a certain edge node access is limited, while the privacy risk it can bear is also upper-bounded. Thus, we introduce two constraints:
\begin{align}
    \label{game:utility resource}
    \sum_{t=0}^{T-1} (B_k^t)^2 \leq n_k, \sum_{t=0}^{T-1} (\varepsilon_k^t)^2 \leq m_k,
\end{align}
where $n_k$ and $m_k$ are the upper bounds of computation resource and privacy risk for edge node $k$, respectively.

In summary, the optimization problem on the edge nodes' side can be formulated as
\begin{align}
    \label{game:utility opti}
    \max_{\boldsymbol B_k, \boldsymbol \varepsilon_k} & U_k(\boldsymbol R, \boldsymbol B_k, \boldsymbol B_{-k}, \boldsymbol \varepsilon_k, \boldsymbol \varepsilon_{-k}) \\
    s.t. & \sum_{t=0}^{T-1} (B_k^t)^2 \leq n_k, \sum_{t=0}^{T-1} (\varepsilon_k^t)^2 \leq m_k. \notag
\end{align}

\subsection{Stackelberg Game Formulation}
Based on the discussion on Sections \ref{sect:1} and \ref{sect:2}, we can find that optimization problems (\ref{game:cost opti}) and (\ref{game:utility opti}) are influenced by each other, which makes it impossible to derive the optimal strategies for the cloud server and edge nodes individually. To formulate the interaction between the two optimization problems, we model them as a two-stage Stackelberg Game:
\begin{align}
    \label{game:stackel}
    \text{Stage  I}: & \min_{\boldsymbol R} C(\boldsymbol R, \boldsymbol B, \boldsymbol \varepsilon); \notag \\
    \text{Stage  II}: & \max_{\boldsymbol B_k, \boldsymbol \varepsilon_k} U_k(\boldsymbol R, \boldsymbol B_k, \boldsymbol B_{-k}, \boldsymbol \varepsilon_k, \boldsymbol \varepsilon_{-k}), \\
    & s.t. \sum_{t=0}^{T-1} (B_k^t)^2 \leq n_k, \sum_{t=0}^{T-1} (\varepsilon_k^t)^2 \leq m_k, \notag
\end{align}
where the cloud server acts as the leader, and edge nodes respond as followers. By deriving the Nash Equilibrium of this game, we can get a set of mutually optimal strategies between the cloud server and edge nodes in a stable condition.
    
\section{Methodology}
In this section, we explore the Nash Equilibrium of the above Stackelberg Game by means of backward reduction. Firstly, we analyze edge node $k$'s optimal strategy $(\boldsymbol B_k, \boldsymbol \varepsilon_k)^*$ in Stage II given any cloud server's payment $\boldsymbol R$. Then we discuss the optimal strategy $\boldsymbol R^*$ based on $\{(\boldsymbol B_k, \boldsymbol \varepsilon_k)^*\}_{k=1}^N$ in Stage I.

\subsection{Optimal Strategy for Edge Nodes}
Before the analysis of the strategy for edge nodes, we face two key challenges:
1) \textit{Incomplete information}. As shown in Equ. (\ref{game:utility}), deriving edge node $k$'s optimal strategy $(\boldsymbol B_k, \boldsymbol \varepsilon_k)^*$ requires the global knowledge of $\sum_{i=1}^N \log (B_i^t \varepsilon_i^t)$ in the game. Yet $\sum_{i=1}^N \log (B_i^t \varepsilon_i^t)$ is usually kept unknown to edge node $k$ due to inter-edge node information isolation in federated learning tasks.
2) \textit{Resource Constraints}. Equ. (\ref{game:utility opti}) is an optimization problem with a long-term objective function and time-average constraints. Strategies made in former slots will affect latter ones, and it is difficult to derive the optimal strategy for previous slots considering unpredictable circumstances in the future.

To cope with the first challenge, we introduce a mean-field estimator $\phi^t$ to approximate $\sum_{i=1}^N \log (B_i^t \varepsilon_i^t)$. Mathematically, $\phi^t$ is a given function and viewed as a known term here. The estimation of $\phi^t$ will be discussed later in Section \ref{met:3}.

By alternating $\sum_{i=1}^N \log (B_i^t \varepsilon_i^t)$ in Equ. (\ref{game:utility opti}) with $\phi^t$, the optimization problem of edge node $k$ is rewritten as
\begin{align}
    \label{method:phi}
    \max_{\boldsymbol B_k, \boldsymbol \varepsilon_k} & \sum_{t=0}^{T-1} \left(\frac{\log(B_k^t\varepsilon_k^t)}{\phi^t} R^t - \alpha_k(B_k^t)^2 - \beta_k(\varepsilon_k^t)^2\right), \notag \\
    s.t. & \sum_{t=0}^{T-1} (B_k^t)^2 \leq n_k, \sum_{t=0}^{T-1} (\varepsilon_k^t)^2 \leq m_k. 
\end{align}

To handle the second challenge, we proposed an online strategy-making approach based on the Lyapunov drift-plus-penalty framework, which transforms the time-average resource constraints in Equ. (\ref{game:utility opti}) into queue stability problems. Specifically, we define virtual queues as
\begin{align}
    Q_k^{t+1} = \max\left\{Q_k^t + (B_k^t)^2 - \frac{n_k}{T}, 0\right\}, \forall k\in[1, N], \label{method:queue1}\\
    Z_k^{t+1} = \max\left\{Z_k^t + (\varepsilon_k^t)^2 - \frac{m_k}{T}, 0\right\}, \forall k\in[1, N], \label{method:queue2}
\end{align}
with initial condition of $Q_k^1 = 0$ and $Z_k^1 = 0$. The above virtual queues capture accumulated violations of resource constraints. By ensuring the stability of virtual queues, we can guarantee the satisfaction of the time-average resource constraints within a bounded violation error.

Using virtual queues $Q_k^t$ and $Z_k^t$, Equ. (\ref{method:phi}) can be further transformed into single-slot optimization problems. For a particular round $t$, the optimization problem for edge node $k$ is rewritten as
\begin{align}
    \label{method:new utility opti}
    \min_{B_k^t, \varepsilon_k^t} & \gamma_2\left(\alpha_k (B_k^t)^2 + \beta_k (\varepsilon_k^t)^2 - \frac{\log(B_k^t \varepsilon_k^t)}{\phi^t} R^t\right) \notag \\
    & + Q_k^t\left((B_k^t)^2 - \frac{n_k}{T}\right) + Z_k^t\left((\varepsilon_k^t)^2 - \frac{m_k}{T}\right).
\end{align}
It targets to optimize edge node $k$'s utility and the queue stability of $Q_k^t, Z_k^t$ simultaneously, with $\gamma_2 > 0$ working as the weight factor. Note that Equ. (\ref{method:new utility opti}) is an online problem because solving it requires the real-time state of virtual queues.

Given mean-field estimator $\phi^t$ and payment $R^t$ launched by the cloud server, the optimal strategy $(B_k^t, \varepsilon_k^t)$ for edge node $k$ at round $t$ is as follows:
\begin{proposition}
\label{the:2}
For any edge node $k$ at arbitrary round $t$, the optimal strategy $(B_k^t, \varepsilon_k^t)^*$ is
\begin{align}
    (B_k^t)^* = & \sqrt{\frac{\gamma_2R^t}{2\phi^t(\gamma_2\alpha_k + Q_k^t)}}, \label{method:edge node answer 1} \\ 
    (\varepsilon_k^t)^* = & \sqrt{\frac{\gamma_2R^t}{2\phi^t(\gamma_2\beta_k + Z_k^t)}}. \label{method:edge node answer 2}
\end{align}
\end{proposition}
The detailed proof is provided in Appendix A.2 in the supplementary material.
Proposition \ref{the:2} uncovers that $(B_k^t, \varepsilon_k^t)^*$ increases with $R^t$, which means a greater payment by the cloud server appeals to edge nodes to risk higher privacy leakage in exchange for economic reward, while high unit cost of $\alpha_k, \beta_k$ and unstable virtual queue of $Q_k^t, Z_k^t$ have the opposite effect. 

\subsection{Optimal Strategy for Cloud Server}
In this section, we explore the optimal strategy $(R^t)^*$ for the cloud server given all edge nodes' strategy $\left\{(B_k^t, \varepsilon_k^t)^*\right\}_{k=1}^N$ at arbitrary round $t$. Based on the backward reduction, we substitute $\left\{(B_k^t, \varepsilon_k^t)^*\right\}_{k=1}^N$ into the cloud server's cost function in Equ. (\ref{game:cost}), and the optimal strategy $R^t$ for the cloud edge under given mean-field estimator $\phi^t$ is as follows:
\begin{proposition}
    \label{the:3}
    The optimal strategy $(R^t)^*$ for cloud server at arbitrary round $t$ is
    \begin{align}
        (R^t)^* = & \left( \frac{2 {\kappa_1}^{T-1-t}\kappa_3\eta^2C^2}{\gamma_1} \cdot \frac{\sum_{k=1}^N (Y_k^t)^{-1}}{\left(\sum_{k=1}^N (X_k^t)^{\frac{1}{2}}\right)^2}\right)^{\frac{1}{3}}, \label{method:cloud server answer} \\
        X_k^t = & \frac{\gamma_2}{2\phi^t(\gamma_2 \alpha_k + Q_k^t)},
        Y_k^t = \frac{\gamma_2}{2\phi^t(\gamma_2\beta_k + Z_k^t)}. \notag
    \end{align}
\end{proposition}

The detailed proof is provided in Appendix A.3 in the supplementary material.
$X_k^t$ and $Y_k^t$ are defined as quality factors of edge node $k$ at round $t$ in term of two resource constraints. We say edge node $k$ is of high quality if it features a lower unit cost $\alpha_k, \beta_k$ and more stable virtual queues $X_k^t, Y_k^t$ at round $t$. According to Proposition \ref{the:3}, the cloud server has to afford more monetary payment for low-quality edge nodes to guarantee full participation of them, which is consistent with our intuition.

\subsection{Algorithm for Finalizing Strategy Design}
\label{met:3}
In this section, we explore finding the precise value of the mean-field estimator $\phi(t)$, thereby finalizing strategy design for the Stackelberg Game. On the one hand, $\phi^t$ defined as $\sum_{k=1}^N \log(B_k^t \varepsilon_k^t)$ is affected by $\left\{(B_k^t, \varepsilon_k^t)\right\}_{k=1}^N$; on the other hand, $\phi^t$ will in turn affect the determination of $\left\{(B_k^t, \varepsilon_k^t)\right\}_{k=1}^N$ according to Proposition \ref{the:2}. There is a closed-loop among $\phi^t$ and $\left\{(B_k^t, \varepsilon_k^t)\right\}_{k=1}^N$. Based on this, we have the following proposition:
\begin{proposition}
    \label{the:4}
    There exists a fixed point for the mean-field estimator $\{\phi^t\}_{t=0}^{T-1}$.
\end{proposition}
The detailed proof is provided in Appendix A.4 in the supplementary material [2].
Based on Proposition \ref{the:4}, we develop a fixed-point approach to determine ${\phi^t}$, which will be introduced later in Section \ref{sect:alg}.

In summary, the Nash Equilibrium for Equ. (\ref{game:stackel}) is
\begin{align}
    \label{method:equilibrium answer}
    \text{Stage  I}: & (R^t)^*, \notag \\
    \text{Stage  II}: & (B_k^t, \varepsilon_k^t)^*.
\end{align}

\subsection{Complete Workflow for FEDBUD Mechanism}
\label{sect:alg}
The complete algorithm of FEDBUD is summarized in Algorithm \ref{alg:1}. Take round $t$ for instance:
\begin{enumerate}
    \item \textit{Strategy Decision Phase}: the system initializes mean-field estimator $\phi_0^t$. Amid the $i$-th fixed-point iteration, given $\phi_i^t$, the cloud server optimizes strategy $R_i^t$ to minimize its cost function before edge nodes optimize strategy $(B_{k, i}^t, \varepsilon_{k, i}^t)$ to maximize their utility function, which is followed by the update of mean-field estimator $\phi_{i+1}^t$. Iterations will come to the end until convergence, when $\phi^t, R^t$ and $\{(B_k^t, \varepsilon_k^t)\}_{k=1}^N$ are fixed synchronously.
    \item \textit{Federated Training Phase}: the cloud server distributes global model $w^t$ with optimal payment $R^t$ to edge nodes. After that, edge node $k$ conducts local training with optimal data volume $B_k^t$ and injects noise according to optimal privacy budget $\varepsilon_k^t$. In addition, edge node $k$ updates virtual queue of $Q_k^{t+1}$ and $Z_k^{t+1}$ for next round's use.
\end{enumerate}
After $T$ rounds of federated training, Algorithm \ref{alg:1} returns global model $\mathbf w^T$.

\begin{algorithm}[H]
\caption{FEDBUD Mechanism}
\label{alg:1}
\begin{algorithmic}[1]
\STATE \textbf{Input:} number of rounds $T$, number of clients $N$.
\STATE \textbf{Output:} global model $\mathbf w^{T}$.
\STATE \textbf{Initialize:} global model $w^0$, virtual queues $\{Q_k^1\}_{k=1}^N$ and $\{Z_k^1\}_{k=1}^N$, other hyperparameters.
\FOR{$t=0$ to $T-1$}
    \STATE \underline{\textbf{Strategy Decision Phase}}
    \STATE \textbf{Initialize:} mean-field estimator $\phi_0^t$, iteration counter $i = 0$.
    \REPEAT
        \STATE Cloud server computes the optimal $R_i^t$ based on $\phi_i^{t}$ according to Proposition \ref{the:3}.
        \FOR{edge node $k=1$ to $N$}
            \STATE Compute $(B_{k,i}^t,\varepsilon_{k,i}^t)$ according to Proposition \ref{the:2}.
        \ENDFOR
        \STATE Update estimator $\phi_{i+1}^{t}\gets \sum_{k=1}^{N}\log(B_k^t\varepsilon_k^t)$.
        \STATE $i\gets i+1$.
    \UNTIL{$|\phi_i^{t}-\phi_{i-1}^{t}|\le \epsilon$}.
    \STATE Set $(\phi^{t}, R^t, B_k^t,\varepsilon_k^t) \gets (\phi_i^t, R_i^t, B_{i,k}^t, \varepsilon_{i,k}^t)$.

    \STATE \underline{\textbf{Federated Training Phase}}
    \STATE Cloud server broadcasts $(\mathbf w^t, R^t)$ to all edge nodes. 
    \FOR{edge node $k=1$ to $N$}
        \STATE Perform local training and noise injecting using $(B_k^t,\varepsilon_k^t)$ according to Equ. (\ref{formu:local}). 
        \STATE Upload local model $\mathbf w_k^{t+1}$ to server.
        \STATE Update virtual queue $Q_k^{t+1}$ according to Equ. (\ref{method:queue1}).
        \STATE Update virtual queue $Z_k^{t+1}$ according to Equ. (\ref{method:queue2}).
    \ENDFOR
    \STATE Cloud server aggregates model according to Equ. (\ref{formu:aggre}).
\ENDFOR
\end{algorithmic}
\end{algorithm}

\section{Experiments}
In this section, we evaluate the performance of our proposed FEDBUD by numerical experiments.

\begin{figure*}[htbp]
  \centering
  \includegraphics[width=\linewidth]{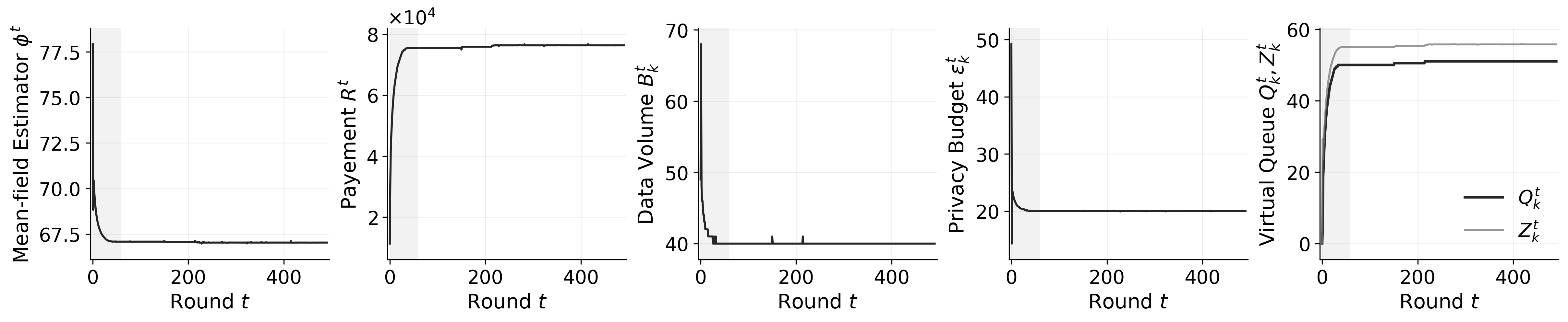}
  \caption{Illustration of movement trajectory for mean-field estimator $\phi^t$, the cloud server’s strategy $R^t$, edge node $k$' strategy $(B_k^t, \varepsilon_k^t)$ and virtual queues $Q_k^t, Z_k^t$ over the time horizon.}
  \vspace{-1em}
  \label{fig:1}
\end{figure*}

\begin{figure}[htbp]
    \centering
    \begin{subfigure}[htbp]{0.49\linewidth}
        \centering
        \includegraphics[width=\linewidth]{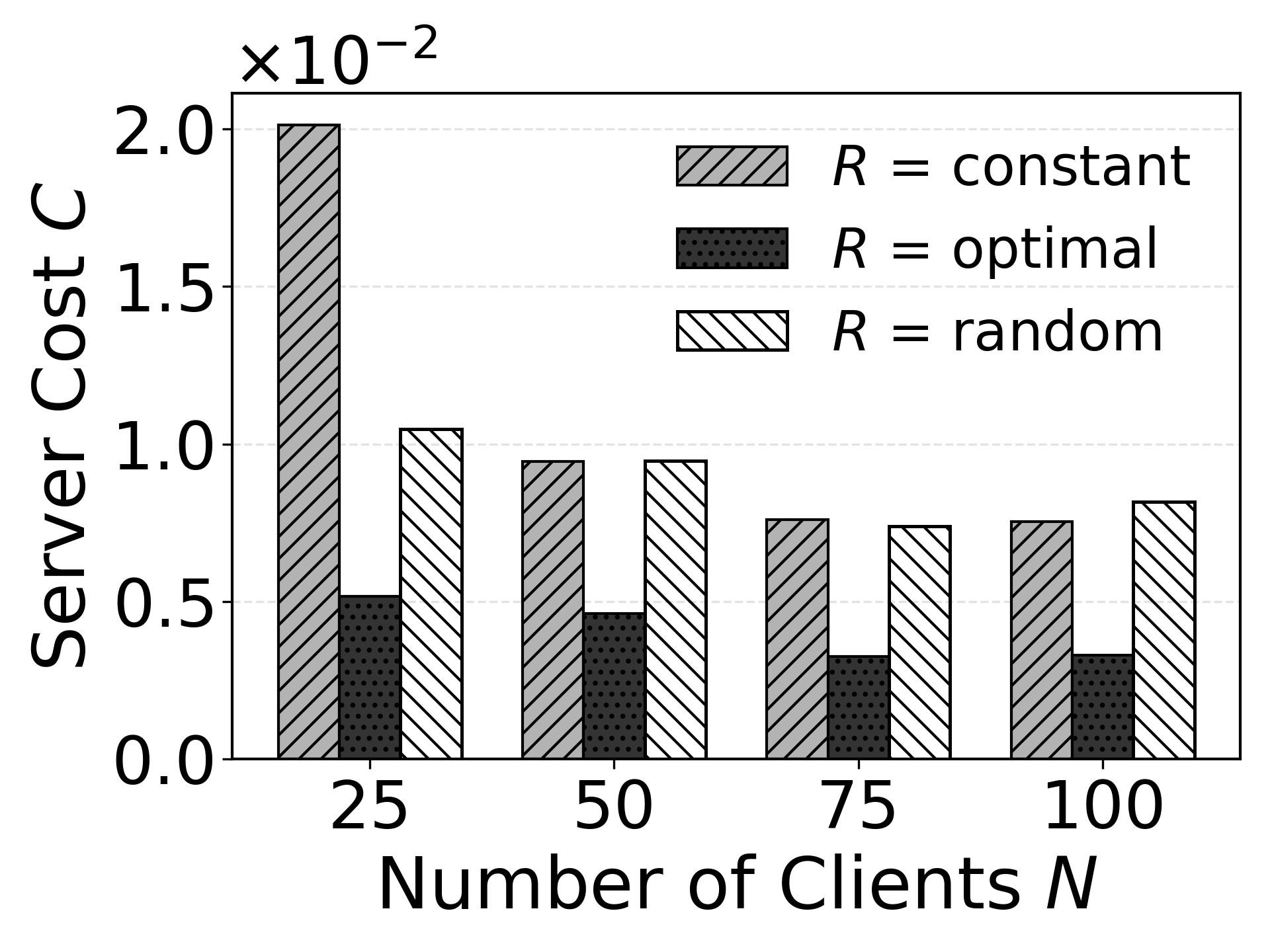}
    \end{subfigure}\hfill
    \begin{subfigure}[htbp]{0.49\linewidth}
        \centering
        \includegraphics[width=\linewidth]{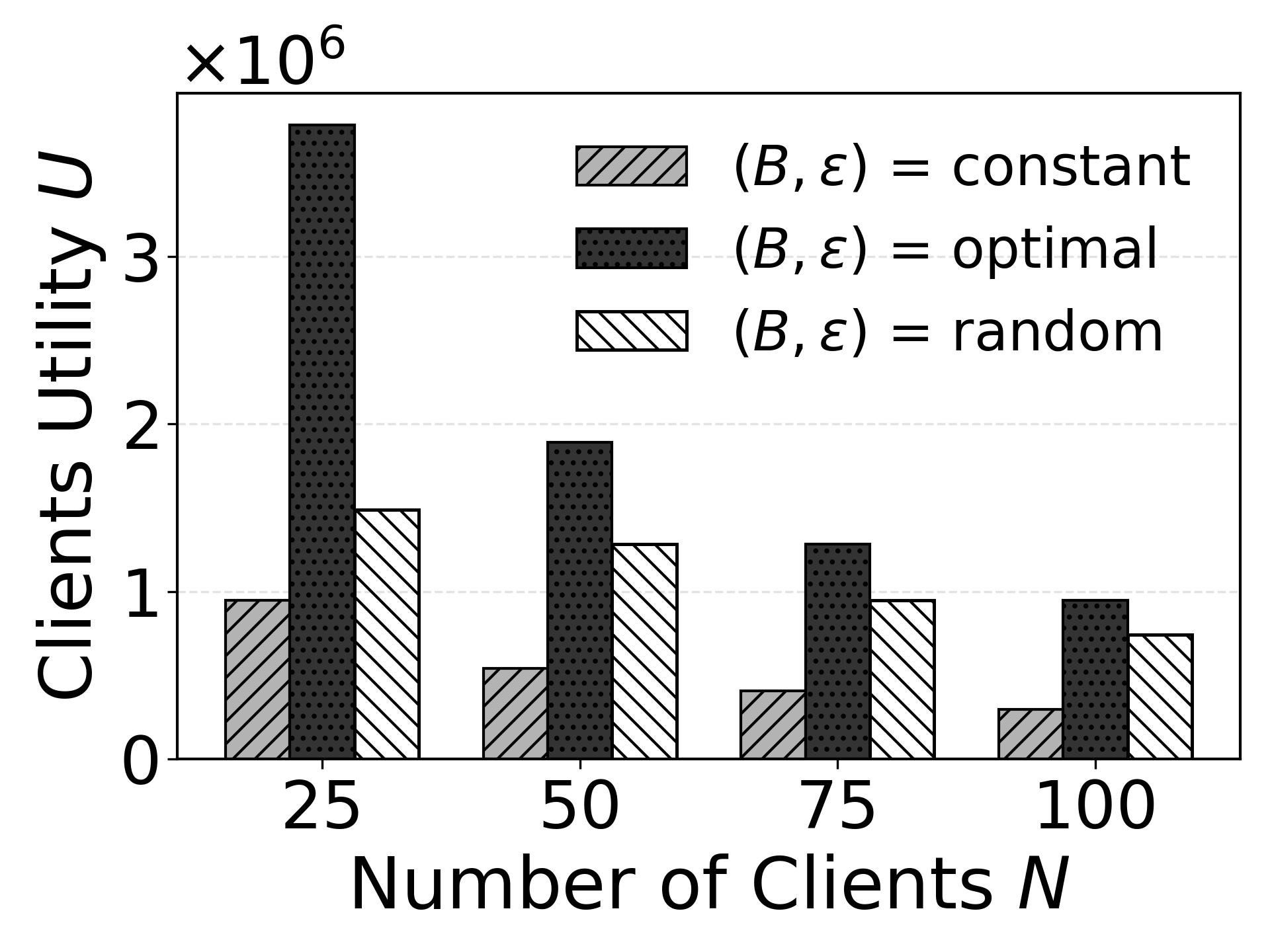}
    \end{subfigure}\hfill
    \caption{Comparison of cloud server's cost $C$ (\textbf{left}) and edge node $k$'s utility $U_k$ (\textbf{right}) over different strategies.}
    \vspace{-1.5em}
    \label{fig:2}
\end{figure}

\begin{table}[t]
\centering
\caption{Impact of weight factor $\gamma_1$ on the cloud server's objective trade-off.}
\label{tab:gamma}
\setlength{\tabcolsep}{5.5pt}
\renewcommand{\arraystretch}{1.08}
\begin{tabular}{c | c c}
\toprule
\makecell{Factor $\gamma_1$} & \makecell{Payment $\sum R^t$} & \makecell{Model Loss $F(\mathbf w^T) - F(\mathbf w^*)$} \\
\midrule
$\phantom{0}1\times10^{-11}$ &  $398.838 \times 10^6$ & $1.337 \times 10^{-3}$\\
$\phantom{0}5\times10^{-11}$ &  $\phantom{0}78.368 \times 10^6$ & $1.347 \times 10^{-3}$\\
$10 \times10^{-11}$ &  $\phantom{0}38.174 \times 10^6$ & $1.354 \times 10^{-3}$\\
\bottomrule
\end{tabular}
\end{table}





\subsection{Settings}
In our experiments, we arrange federated tasks on the widely used benchmark of CIFAR-10. We set $T = 100$ communication rounds, and $N = 100$ edge nodes participate. Each edge node conducts local update using Stochastic Gradient Descent (SGD) with a learning rate $\eta = 10^{-3}$ for 10 epochs. The unit cost for computation resource and privacy risk obeys $\alpha_k\sim\mathcal{U}(1 \times 10^{-2}, 5 \times 10^{-2}), \beta_k\sim \mathcal{U}(1\times 10^{-2}, 5\times 10^{-2})$. The weight factor is set as $\gamma_1 = 1 \times 10^{-10}, \gamma_2 = 1$. All experiments are implemented in PyTorch and conducted on a workstation equipped with an NVIDIA GPU. To accelerate training, multiple clients are executed in parallel using multiprocessing.

\subsection{Performance Evaluation}
We illustrate the performance evaluation of FEDBUD in this section. 

\textbf{Fixed-point Convergence Process: } Fig. \ref{fig:1} shows the movement trajectory of mean-field estimator $\phi^t$, the cloud server's strategy $R^t$, edge node $k$' strategy $(R_k^t, \varepsilon_k^t)$ and virtual queues $Q_k^t, Z_k^t$ over the time horizon. In the early stage, the system is in an unstable condition where both the cloud server and edge nodes are exploring their optimal strategies. After only 50-60 rounds, both sides fix their best or near-best $((R^t)^* \rightarrow 7.63 \times 10^4, (B_k^t, \varepsilon_k^t)^* \rightarrow (40.1, 20.7))$, which means the system converges to a relatively stable state for model training. The result shows the viability and efficiency of Algorithm \ref{alg:1} in solving the complex problem of Equ. (\ref{game:stackel}). In addition, we can find virtual queue $Q_k^t$ and $Z_k^t$ keep within a limited upper bound of $50 - 60$ throughout the task. Despite sight resource infringement, the general stability of virtual queues demonstrates that strategies derived by Algorithm \ref{alg:1} obey resource constraints literally.

\textbf{Verification of Derived Strategy as Nash Equilibrium: } In this paragraph, we verify the strategy of Equ. (\ref{method:equilibrium answer}) as Nash Equilibrium. For comparison, we set two auxiliary strategies: 1) \textit{Constant strategy}. It means the cloud server or edge nodes take static actions over the time horizon, with the value not equal to the converged results in Fig. \ref{fig:1} $(R^{constant} \neq 7.63 \times 10^4; (B_k, \varepsilon_k)^{constant} \neq (40.1, 20.7))$. 2) \textit{Random strategy}. It refers to taking random actions over the task. Considering fairness, the slot-average value under this strategies is set to keep in line with that of converged results in Fig. \ref{fig:1} (randomized $R^{random}$ with $\frac{1}{T} \sum_{t=0}^{T-1} R^{random} = 7.63 \times 10^4$; randomized $(B_k, \varepsilon_k)^{random}$ with $\frac{1}{T} \sum_{t=0}^{T-1} B_k^{random} = 40.1, \frac{1}{T}\sum_{t=0}^{T-1} \varepsilon_k^{random} = 20.7$).
As plotted in Fig. \ref{fig:2}, compared with other baselines, $(R^t)^*$ helps the cloud server obtain the lowest cost, while $(B_k^t, \varepsilon_k^t)^*$ helps edge nodes secure the highest utility. Provided that both the cloud server and edge nodes are selfish, the above results mean they will strictly obey the optimal strategies derived by Algorithm \ref{alg:1} rather than others, thereby the mutually optimal strategies are reached simultaneously, and the Nash Equilibrium holds.
In addition, we can find that both the cloud server's cost and edge nodes' objectives decrease marginally with the number of edge nodes $N$. For edge nodes, number expansion intensifies competition for payment, further leading to allocated payment reduction and utility reduction. For the cloud server, despite incurring more payment, numerous edge nodes help improve model performance in return, which reduces the overall cost.

\textbf{Impact of Weight Factor on Strategy: } In this paragraph, we explore the impact of weight factor $\gamma_1$ on the cloud server's strategy. For comparison, we set a range of $\gamma_1$ (from 1 to 10, $\times 10^{-11}$) and the results are plotted in Table \ref{tab:gamma}. It shows training loss $F(\mathbf w^T) - F(\mathbf w^*)$ increases (from 1.337 to 1.354, $\times 10^{-3}$) with $\gamma_1$ while the accumulated monetary payment $\sum_{t=1}^T R^t$ decreases (from 398.838 to 38.174, $\times 10^6$) with it. The underlying reason is that the cloud server takes priority to economic expenditure rather than model performance with the growth of $\gamma_1$. In addition, the payment-loss relationship is drastically nonlinear. Once model performance approaches a near-saturated state, further reducing $\gamma_1$ will result in tremendous payment in exchange for redundant contribution to model performance enhancement. Therefore, it is an important step to determine $\gamma_1$ according to real-world demand in Algorithm \ref{alg:1}.

\section{Conclusion}
In this paper, we propose a novel federated learning system called FEDBUD, which combines privacy and economic concerns together by considering the joint influence of data volume and noise level on incentive strategy determination. To determine the optimal strategies for both sides, we model FEDBUD as a two-stage Stackelberg Game and derive the Nash Equilibrium. 
Extensive experiments demonstrate the superiority of our proposed approach.


\bibliography{root}

@article{nasajpour2025federated,
  title={Federated Learning in Smart Healthcare: A Survey of Applications, Challenges, and Future Directions},
  author={Nasajpour, Mohammad and Pouriyeh, Seyedamin and Parizi, Reza M and Han, Meng and Mosaiyebzadeh, Fatemeh and Liu, Liyuan and Xie, Yixin and Batista, Daniel Mac{\^e}do},
  journal={Electronics},
  volume={14},
  number={9},
  pages={1750},
  year={2025},
  publisher={MDPI}
}

@article{gandhi2025federated,
  title={Federated Learning in Secure Smart City Sensing: Challenges and Opportunities},
  author={Gandhi, Monika and Singh, Sushil Kumar and Ravikumar, RN and Vaghela, Krunal},
  journal={Edge of Intelligence: Exploring the Frontiers of AI at the Edge},
  pages={215--251},
  year={2025},
  publisher={Wiley Online Library}
}

@inproceedings{yang2015bayesian,
  title={Bayesian differential privacy on correlated data},
  author={Yang, Bin and Sato, Issei and Nakagawa, Hiroshi},
  booktitle={Proceedings of the 2015 ACM SIGMOD international conference on Management of Data},
  pages={747--762},
  year={2015}
}

@inproceedings{triastcyn2020bayesian,
  title={Bayesian differential privacy for machine learning},
  author={Triastcyn, Aleksei and Faltings, Boi},
  booktitle={International Conference on Machine Learning},
  pages={9583--9592},
  year={2020},
  organization={PMLR}
}

@article{zhang2023age,
  title={Age-dependent differential privacy},
  author={Zhang, Meng and Wei, Ermin and Berry, Randall and Huang, Jianwei},
  journal={IEEE Transactions on Information Theory},
  volume={70},
  number={2},
  pages={1300--1319},
  year={2023},
  publisher={IEEE}
}

@inproceedings{lin2024age,
  title={Age aware scheduling for differentially-private federated learning},
  author={Lin, Kuan-Yu and Lin, Hsuan-Yin and Hsu, Yu-Pin and Huang, Yu-Chih},
  booktitle={2024 IEEE International Symposium on Information Theory (ISIT)},
  pages={398--403},
  year={2024},
  organization={IEEE}
}

@article{huang2024collaboration,
  title={Collaboration in federated learning with differential privacy: A stackelberg game analysis},
  author={Huang, Guangjing and Wu, Qiong and Sun, Peng and Ma, Qian and Chen, Xu},
  journal={IEEE Transactions on Parallel and Distributed Systems},
  volume={35},
  number={3},
  pages={455--469},
  year={2024},
  publisher={IEEE}
}

@article{tang2025game,
  title={Game-Theoretic Incentive Mechanism for Blockchain-Based Federated Learning},
  author={Tang, Wenzheng and Liu, Erwu and Ni, Wei and Qu, Xinyu and Huang, Butian and Li, Kezhi and Niyato, Dusit and Jamalipour, Abbas},
  journal={IEEE Transactions on Mobile Computing},
  year={2025},
  publisher={IEEE}
}

@inproceedings{chen2025dualgfl,
  title={DualGFL: Federated Learning with a Dual-Level Coalition-Auction Game},
  author={Chen, Xiaobing and Zhou, Xiangwei and Zhang, Songyang and Sun, Mingxuan},
  booktitle={Proceedings of the AAAI Conference on Artificial Intelligence},
  volume={39},
  number={15},
  pages={15904--15912},
  year={2025}
}

@inproceedings{tang2025reputation,
  title={Reputation-aware Revenue Allocation for Auction-based Federated Learning},
  author={Tang, Xiaoli and Yu, Han},
  booktitle={Proceedings of the AAAI Conference on Artificial Intelligence},
  volume={39},
  number={19},
  pages={20832--20840},
  year={2025}
}

@article{xie2025privacy,
  title={A Privacy-Preserving Incentive Scheme for UAV-Aided Federated Learning: A Contract Method With Prospect Theory},
  author={Xie, Liang and Su, Zhou and Wang, Yuntao and Chen, Nan and Liu, Yiliang and Wang, Rui and Liu, Xin and Liu, Donglan and Zhang, Hao},
  journal={IEEE Transactions on Dependable and Secure Computing},
  year={2025},
  publisher={IEEE}
}

@inproceedings{wang2023federated,
  title={Federated learning with flexible control},
  author={Wang, Shiqiang and Perazzone, Jake and Ji, Mingyue and Chan, Kevin S},
  booktitle={IEEE INFOCOM 2023-IEEE Conference on Computer Communications},
  pages={1--10},
  year={2023},
  organization={IEEE}
}

@article{luo2024adaptive,
  title={Adaptive heterogeneous client sampling for federated learning over wireless networks},
  author={Luo, Bing and Xiao, Wenli and Wang, Shiqiang and Huang, Jianwei and Tassiulas, Leandros},
  journal={IEEE Transactions on Mobile Computing},
  volume={23},
  number={10},
  pages={9663--9677},
  year={2024},
  publisher={IEEE}
}

@article{zhan2019free,
  title={Free market of multi-leader multi-follower mobile crowdsensing: An incentive mechanism design by deep reinforcement learning},
  author={Zhan, Yufeng and Liu, Chi Harold and Zhao, Yinuo and Zhang, Jiang and Tang, Jian},
  journal={IEEE Transactions on Mobile Computing},
  volume={19},
  number={10},
  pages={2316--2329},
  year={2019},
  publisher={IEEE}
}

@article{nie2020multi,
  title={A multi-leader multi-follower game-based analysis for incentive mechanisms in socially-aware mobile crowdsensing},
  author={Nie, Jiangtian and Luo, Jun and Xiong, Zehui and Niyato, Dusit and Wang, Ping and Poor, H Vincent},
  journal={IEEE Transactions on Wireless Communications},
  volume={20},
  number={3},
  pages={1457--1471},
  year={2020},
  publisher={IEEE}
}
\bibliographystyle{IEEEtran}

\newpage
\appendices

In the appendix, the complete proofs of theoretic results provided in the main text are exhibited in detail.

\section{Proof of Theorem \ref{the:1}}

\begin{proof}
    Let us pay attention to round $t + 1$. According to Assumption \ref{ass:2}, we have
    \begin{align}
        \label{R:1}
        & E\left[F(\mathbf w^{t+1}) - F(\mathbf w^t)\right] \notag \\
        \leq & \underbrace{E\left\langle \nabla F(\mathbf w^t), \mathbf w^{t+1} - \mathbf w^t\right\rangle \vphantom{\frac{\rho}{2}}}_A + \underbrace{\frac{\rho}{2} E\left\Vert \mathbf w^{t+1} - \mathbf w^t \right\Vert_2^2}_B.
    \end{align}
    First, we focus on bounding $A$:
    \begin{align}
        & E\left\langle \nabla F(\mathbf w^t), \mathbf w^{t+1} - \mathbf w^t\right\rangle \notag \\
        = & E\left\langle \nabla F(\mathbf w^t), -\eta \nabla F(\mathbf w^t) + \mathbf n^t \right\rangle \notag \\
        = & E\left\langle \nabla F(\mathbf w^t), -\eta \nabla F(\mathbf w^t) \right\rangle + E\left\langle \nabla F(\mathbf w^t), \mathbf n^t \right\rangle \notag \\
        = & (-\eta) E\left\Vert \nabla F(\mathbf w^t)\right\Vert_2^2.
    \end{align}
    The third step holds due to zero-mean noise in DP.
    Then, we focus on bounding $B$:
    \begin{align}
        & \frac{\rho}{2} E\left\Vert \mathbf w^{t+1} - \mathbf w^t \right\Vert_2^2 \notag \\
        = &\frac{\rho}{2} E\left\Vert \sum_{k=1}^N \frac{B_k^t}{B^t} (-\eta \nabla F_k(\mathbf w^t; \mathcal{B}_k^t) + \mathbf n_k^t)\right\Vert_2^2 \notag \\
        = &\frac{\rho}{2} E\left\Vert \sum_{k=1}^N\frac{B_k^t}{B^t} (-\eta\nabla F_k(\mathbf w^t; \mathcal{B}_k^t)) \right\Vert_2^2 + \frac{\rho}{2} E\left\Vert \sum_{k=1}^N \frac{B_k^t}{B^t} \mathbf n_k^t \right\Vert_2^2 \notag \\
        & - \rho E\left\langle \sum_{k=1}^N\frac{B_k^t}{B^t}(-\eta\nabla F_k(\mathbf w^t); \mathcal{B}_k^t)), \sum_{k=1}^N \frac{B_k^t}{B^t} \mathbf n_k^t \right\rangle \notag \\
        = &\frac{\rho}{2} E\left\Vert \sum_{k=1}^N\frac{B_k^t}{B^t} (-\eta\nabla F_k(\mathbf w^t; \mathcal{B}_k^t))) \right\Vert_2^2 + \frac{\rho}{2} E\left\Vert \sum_{k=1}^N \frac{B_k^t}{B^t} \mathbf n_k^t \right\Vert_2^2 \notag \\
        = &\underbrace{\frac{\rho\eta^2}{2} E\left\Vert \sum_{k=1}^N\frac{B_k^t}{B^t} \nabla F_k(\mathbf w^t; \mathcal{B}_k^t)) \right\Vert_2^2}_{B_1} + \underbrace{\frac{\rho}{2} E\left\Vert \sum_{k=1}^N \frac{B_k^t}{B^t} \mathbf n_k^t \right\Vert_2^2}_{B_2}.
    \end{align}
    According to Assumption \ref{ass:3}, $B_1$ is bounded by
    {\footnotesize
    \begin{align}
        & \frac{\rho\eta^2}{2} E\left\Vert \sum_{k=1}^N\frac{B_k^t}{B^t} \nabla F_k(\mathbf w^t; \mathcal{B}_k^t) \right\Vert_2^2 \notag \\
        \leq & \frac{\rho\eta^2}{2} \sum_{k=1}^N \frac{B_k^t}{B^t} E\left\Vert \nabla F_k(\mathbf w^t; \mathcal{B}_k^t)\right\Vert_2^2 \notag \\
        = & \frac{\rho\eta^2}{2} \sum_{k=1}^N \frac{B_k^t}{B^t} E\left\Vert \nabla F_k(\mathbf w^t; \mathcal{B}_k^t) - F(\mathbf w^t) + F(\mathbf w^t)) \right\Vert_2^2 \notag \\
        \leq & \frac{\rho\eta^2}{2} \sum_{k=1}^N \frac{B_k^t}{B^t} \left( 2 E\left\Vert \nabla F_k(\mathbf w^t; \mathcal{B}_k^t) - F(\mathbf w^t)\Vert_2^2 + 2E\Vert F(\mathbf w^t))\right\Vert_2^2 \right) \notag \\
        \leq & \rho \eta^2 E\left\Vert F(\mathbf w^t) \right\Vert_2^2 + \rho\eta^2\sum_{k=1}^N \frac{B_k^t}{B^t} \lambda_k^t.
    \end{align}
    }
    According to Assumption \ref{ass:4}, $B_2$ is bounded by
    {\footnotesize
    \begin{align}
        & \frac{\rho}{2} E\left\Vert \sum_{k=1}^N \frac{B_k^t}{B^t} \mathbf n_k^t\right\Vert_2^2 
        \leq \frac{\rho}{2}\sum_{k=1}^N \frac{B_k^t}{B^t} E\left\Vert \mathbf n_k^t\right\Vert_2^2
        = \frac{\rho d}{2} \sum_{k=1}^N \left(\frac{B_k^t}{B^t}\right)^2 {\sigma_k^t}^2 \notag \\
        = & \frac{\rho d}{2} \sum_{k=1}^N \left(\frac{B_k^t}{B^t}\right)^2 \left( \frac{\eta C}{B_k^t \varepsilon_k^t}\right)^2
        = \frac{\rho d}{2} \sum_{k=1}^N \frac{\eta^2 C^2}{{B^t}^2 {\varepsilon_k^t}^2}.
    \end{align}
    }

    Combining $A, B_1$ and $B_2$, we have
    \begin{align}
        \label{R:2}
        & E\left[F(\mathbf w^{t+1}) - F(\mathbf w^t)\right] \notag \\
        \leq &\underbrace{(\rho \eta^2 - \eta) E\left\Vert F(\mathbf w^t) \right\Vert_2^2}_{C}
        + \rho\eta^2\sum_{k=1}^N \frac{B_k^t}{B^t} \lambda_k^t + \frac{\rho d}{2} \sum_{k=1}^N \frac{\eta^2 C^2}{{B^t}^2 {\varepsilon_k^t}^2}.
    \end{align}

    Now we bound $C$. We set $\eta < \frac{1}{\rho}$, then $\rho\eta^2 - \eta < 0$. According to Assumption \ref{ass:2}, the following inequality holds:
    \begin{align}
        \label{C:2}
        & 2\mu (\rho \eta^2 - \eta) E[F(\mathbf w^t) - F(\mathbf w^*)]
        \geq (\rho \eta^2 - \eta) E\left\Vert \Delta F(\mathbf w^t) \right\Vert_2^2.
    \end{align}

    Substituting Equ. (\ref{C:2}) into Equ. (\ref{R:2}), we have
    \begin{align}
        \label{R:3}
        & E\left[F(\mathbf w^{t+1}) - F(\mathbf w^t)\right] \notag \\
        \leq & 2\mu (\rho \eta^2 - \eta) E[F(\mathbf w^t) - F(\mathbf w^*)] \notag \\
        &+ \rho\eta^2\sum_{k=1}^N \frac{B_k^t}{B^t} \lambda_k^t + \frac{\rho d}{2} \sum_{k=1}^N \frac{\eta^2 C^2}{{B^t}^2 {\varepsilon_k^t}^2}.
    \end{align}

    Adding $E[F(\mathbf w^{t-1}) - F(\mathbf w^*)]$ on both sides on Equ. (\ref{R:3}), we have
    \begin{align}
        \label{R:4}
        & E\left[F(\mathbf w^{t+1}) - F(\mathbf w^*)\right] \notag \\
        \leq & (1 + 2\mu\rho \eta^2 - 2\mu\eta) E[F(\mathbf w^t) - F(\mathbf w^*)] \notag \\
        &+ \rho\eta^2\sum_{k=1}^N \frac{B_k^t}{B^t} \lambda_k^t + \frac{\rho d}{2} \sum_{k=1}^N \frac{\eta^2 C^2}{{B^t}^2 {\varepsilon_k^t}^2}.
    \end{align}

    Recursively using Equ. (\ref{R:4}), we have
    \begin{align}
        \label{R:5}
        & E\left[F(\mathbf w^{t+1}) - F(\mathbf w^*)\right] \notag \\
        \leq & (1 + 2\mu\rho \eta^2 - 2\mu\eta) E[F(\mathbf w^t) - F(\mathbf w^*)] \notag \\
        &+ \rho\eta^2\sum_{k=1}^N \frac{B_k^t}{B^t} \lambda_k^t + \frac{\rho d}{2} \sum_{k=1}^N \frac{\eta^2 C^2}{{B^t}^2 {\varepsilon_k^t}^2} \notag \\
        \leq & (1 + 2\mu\rho \eta^2 - 2\mu\eta)^2 E[F(\mathbf w^{t-1}) - F(\mathbf w^*)] \notag \\
        &+ (1 + 2\mu\rho \eta^2 - 2\mu\eta) \times \notag \\
        &\left(\rho\eta^2\sum_{k=1}^N \frac{B_k^{t-1}}{B^{t-1}} \lambda_k^{t-1} + \frac{\rho d}{2} \sum_{k=1}^N \frac{\eta^2 C^2}{{B^{t-1}}^2 {\varepsilon_k^{t-1}}^2}\right) \notag \\
        &+\left(\rho\eta^2\sum_{k=1}^N \frac{B_k^t}{B^t} \lambda_k^t + \frac{\rho d}{2} \sum_{k=1}^N \frac{\eta^2 C^2}{{B^t}^2 {\varepsilon_k^t}^2}\right) \notag \\
        \leq & \cdots \notag \\
        \leq & (1 + 2\mu\rho \eta^2 - 2\mu\eta)^{t+1} E[F(\mathbf w^0) - F(\mathbf w^*)] \notag \\
        &+ \sum_{r=0}^t (1 + 2\mu\rho \eta^2 - 2\mu\eta)^r \times \notag \\
        & \left(\rho\eta^2\sum_{k=1}^N \frac{B_k^{t-r}}{B^{t-r}} \lambda_k^{t-r} + \frac{\rho d}{2} \sum_{k=1}^N \frac{\eta^2 C^2}{{B^{t-r}}^2 {\varepsilon_k^{t-r}}^2}\right).
    \end{align}

    For ease of representation, let $\kappa_1 = 1 + 2\mu\rho \eta^2 - 2\mu\eta, \kappa_2 = \rho\eta^2, \kappa_3 = \frac{\rho d}{2}$. Thus, the convergence upper bound of Equ. (\ref{R:5}) after $T + 1$ rounds can be formulated as
    \begin{align}
        & E\left[F(\mathbf w^{T+1}) - F(\mathbf w^*)\right] \notag \\
        \leq & {\kappa_1}^{T+1} E[F(\mathbf w^0) - F(\mathbf w^*)] \notag \\
        &+ \sum_{t=0}^T {\kappa_1}^t \left(\kappa_2\sum_{k=1}^N \frac{B_k^{T-t}}{B^{T-t}} \lambda_k^{T-t} + 
        \kappa_3 \sum_{k=1}^N \frac{\eta^2 C^2}{{B^{T-t}}^2 {\varepsilon_k^{T-t}}^2}\right) \notag \\
        \leq & {\kappa_1}^{T+1} E[F(\mathbf w^0) - F(\mathbf w^*)] \notag \\
        &+ \sum_{t=0}^{T} {\kappa_1}^{T-t} \left(\kappa_2\sum_{k=1}^N \frac{B_k^t}{B^t} \lambda_k^t + 
        \kappa_3 \sum_{k=1}^N \frac{\eta^2 C^2}{{B^t}^2 {\varepsilon_k^t}^2}\right).
    \end{align}

    Further, the convergence upper bound after $T$ rounds is
    \begin{align}
        & E[F(\mathbf w^T) - F(\mathbf w^*)] \notag \\
        \leq & {\kappa_1}^T E[F(\mathbf w^0) - F(\mathbf w^*)] \notag \\
        &+ \sum_{t=0}^{T-1} {\kappa_1}^{T-1-t} \left(\kappa_2\sum_{k=1}^N \frac{B_k^t}{B^t} \lambda_k^t + 
        \kappa_3 \sum_{k=1}^N \frac{\eta^2 C^2}{{B^t}^2 {\varepsilon_k^t}^2}\right).
    \end{align}
\end{proof}

\section{Proof of Proposition \ref{the:2}}
\begin{proof}
    According to Equ. (\ref{method:new utility opti}), we have 
    \begin{align}
        f(B_k^t) = \gamma_2(\alpha_k + Q_k^t) (B_k^t)^2 - \frac{R^t}{\phi^t} \log (B_k^t).
    \end{align}

    Then we get the first derivative of $f(B_k^t)$ by
    \begin{align}
        f^{'}(B_k^t) = 2(\alpha_k + Q_k^t) B_k^t - \frac{R^t}{\phi^t} \cdot \frac{1}{B_k^t}.
    \end{align}

    Let $f^{'}(B_k^t) = 0$, we have
    \begin{align}
        (B_k^t)^* = \sqrt{\frac{R^t}{2\phi^t(\alpha_k + Q_k^t)}}.
    \end{align}

    Afterwards, we get the second derivative of $f(B_k^t)$ by
    \begin{align}
        f^{''}(B_k^t) = 2(\alpha_k + Q_k^t) + \frac{R^t}{\phi^t} \cdot \frac{1}{(B_k^t)^2} > 0.
    \end{align}

    If $\alpha_k + Q_k^t > 0$ and $R^t > 0$ hold, $(B_k^t)$ is the optimal solution to minimize Equ. (\ref{method:new utility opti}). The proof for $(\varepsilon_k^t)^*$ is analogous and thus omitted.
\end{proof}

\section{Proof of Proposition \ref{the:3}}
\begin{proof}
    We set
    \begin{equation}
        X_k^t = \frac{\gamma_2}{2\phi^t(\gamma_2 \alpha_k + Q_k^t)}, 
        Y_k^t = \frac{\gamma_2}{2\phi^t(\gamma_2\beta_k + Z_k^t)}.
    \end{equation}
    
    Proposition \ref{the:2} can be reformulated as
    \begin{align}
        (B_k^t)^* = \sqrt{R^t X_k^t},
        (\varepsilon_k^t)^* = \sqrt{R^tY_k^t}.
    \end{align}

    Afterwards, we have
    \begin{align}
        \label{t:1}
        \frac{1}{(B^t)^2} = &  \frac{1}{R^t\left( \sum_{i=1}^N (X_i^t)^{\frac{1}{2}} \right)^2}, 
        \frac{1}{(\varepsilon_k^t)^2} = \frac{1}{R^tY_k^t}. 
    \end{align}

    Substituting Equ. (\ref{t:1}) into Equ. (\ref{game:cost}), we have
    \begin{align}
        f(R^t) 
        = & \gamma_1 R^t + \sum_{k=1}^N \frac{{\kappa_1}^{T-1-t}\kappa_3\eta^2C^2}{{B^t}^2{\varepsilon_k^t}^2} \notag \\
        = &\gamma_1R^t + \sum_{k=1}^N \cdot \frac{{\kappa_1}^{T-1-t}\kappa_3\eta^2C^2}{R^t\left( \sum_{i=1}^N (X_i^t)^{\frac{1}{2}} \right)^2} \cdot \frac{1}{R^tY_k^t} \notag \\
        = & \gamma_1 R^t + \frac{{\kappa_1}^{T-1-t}\kappa_3 \eta^2C^2}{(R^t)^2} \cdot \frac{\sum_{k=1}^N(Y_k^t)^{-1}}{\left( \sum_{i=1}^N (X_i^t)^{\frac{1}{2}} \right)^2}.
    \end{align}
    
    Then we get the first derivative of $f(R^t)$ by
    \begin{align}
        f^{'}(R^t) = \gamma_1 - \frac{2}{(R^t)^3} \cdot {\kappa_1}^{T-1-t}\kappa_3 \eta^2C^2 \cdot \frac{\sum_{k=1}^N (Y_k^t)^{-1}}{\left( \sum_{i=1}^N (X_i^t)^{\frac{1}{2}} \right)^2}.
    \end{align}
    
    Let $f^{'}(R^t) = 0$, we have
    \begin{align}
        (R^t)^* = \left(\frac{2{\kappa_1}^{T-1-t}\kappa_3 \eta^2C^2}{\gamma_1} \cdot \frac{\sum_{k=1}^N (Y_k^t)^{-1}}{\left( \sum_{i=1}^N (X_i^t)^{\frac{1}{2}} \right)^2} \cdot \right)^{\frac{1}{3}}.
    \end{align}
\end{proof}

\section{Proof of Proposition \ref{the:4}}
\begin{proof}
    According to the definition of the mean-field estimator $\phi^t$, we have
    \begin{align}
        \label{fix:0}
        \phi^t = \sum_{k=1}^N\log (B_k^t \varepsilon_k^t),
    \end{align}
    For ease of reading, we rewrite Equ. (\ref{fix:0}) as
    \begin{align}
        \label{fix:1}
        \phi^t = \Psi_1 (B_1^t, \varepsilon_1^t, B_2^t, \varepsilon_2^t, \cdots, B_N^t, \varepsilon_N^t).
    \end{align}
    By inserting Equ. (\ref{method:edge node answer 1}) and Equ. (\ref{method:edge node answer 2}) of Proposition \ref{the:2} into Equ. (\ref{fix:1}), we have
    \begin{align}
        \label{fix:2}
        \phi^t = \Psi_2(\phi^t, R^t),
    \end{align}
    where $\phi^t$ is a function of $(\phi^t, R^t)$.
    Further, by inserting Equ. (\ref{method:cloud server answer}) of Proposition \ref{the:3} into Equ. (\ref{fix:2}), we have
    \begin{align}
        \phi^t = \Psi_4(\phi^t),
    \end{align}
    where $\phi^t$ is literally a function of itself.
\end{proof}

Next, we examine whether a fix point exists for $\Psi_4$. We bound $\phi^t$ as $[0, C]$. On the one hand, $\phi^t >= 0$ holds when $B_k^t\cdot\varepsilon_k^t >= 1$ for all edge nodes $k \in [1, N]$, which is a common assumption in practice \cite{huang2024collaboration}. On the other hand, $\phi^t <= C$ holds since data volume and privacy budget of an edge node are limited according to Equ. (\ref{game:utility resource}).
In general, the domain and range of $\Psi_4$ can be bounded as $\Pi = [0, C]$.

Since $\Psi$ is a continuous mapping from $\Pi$ to $\Pi$, according to Brouwer's fixed-point theorem, $\Psi_4$ has a fix point in $\Pi$ for $\phi^t$.
\end{document}